\documentclass{svmult} 
    
\usepackage{mathptmx}   
\usepackage{helvet} 
\usepackage{courier}
\usepackage{amsmath,amsfonts}
\usepackage{graphicx}
\usepackage[bottom]{footmisc}          
    
\usepackage{mathptmx}         
\usepackage{helvet}            
\usepackage{courier}            
\usepackage{type1cm}             
\usepackage{makeidx}             
\usepackage{graphicx}          
\usepackage{multicol}          
\usepackage[bottom]{footmisc}  
     
\usepackage{epsfig}     
\usepackage{psfrag,rotating}     
\usepackage{amssymb,floatflt,enumerate} 
\usepackage{amsmath,amscd,psfrag,leqno} 
\usepackage[mathscr]{eucal}  
\usepackage{shadethm}           
\usepackage{graphicx,subfigure}   
\usepackage{overpic,contour}       
\usepackage{epstopdf} 
\contourlength{0.3mm}  
 
\def\Beweisende{\square}            
\def\BewEnde{\hfill{\Beweisende}}

\def\phi{\varphi}

\def\RR{{\mathbb R}}
\def\HH{{\mathbb H}}
\def\EE{{\mathbb E}}

\def\FF{{\mathbb F}}

\def\dach#1{\widehat{#1}}

\def\Vkt#1{{\mathbf #1}}

\newcommand{\go}[1]{{\sf #1}}

\makeindex
\begin{document}

\title*{Kinematic interpretation of the 
Study quadric's ambient space}
\author{G. Nawratil}
\authorrunning{G. Nawratil}
%\institute{$^1$IRCCyN, CNRS, France, \email{Philippe.Wenger@irccyn.ec-nantes.fr} \\
%$^1$University of Minho, Portugal, \email{pflores@dem.uminho.pt}}
\institute{
  Institute of Discrete Mathematics and Geometry, Vienna University of Technology, Austria, 
  \email{nawratil@geometrie.tuwien.ac.at}}

\maketitle

\abstract{It is well known that real points of the Study quadric (sliced along a 3-dimensional generator space) 
correspond to displacements of the Euclidean 3-space. But we still lack of a kinematic meaning for the 
points of the ambient 7-dimensional projective space $P^7$. This paper gives one possible interpretation 
in terms of displacements of the Euclidean 4-space. From this point of view we also discuss 
the extended inverse kinematic map, motions corresponding to straight lines in $P^7$ and linear complexes of SE(3)-displacements. 
Moreover we present an application of this interpretation in the context of interactive motion design.
}


\section{Introduction}\label{sec:intro}

$\frak{Q}:=q_0+q_1\Vkt i+q_2\Vkt j+q_3\Vkt k$ with $q_0,\ldots,q_3\in\RR$ is an element of the skew field of quaternions $\HH$, 
where $\Vkt i,\Vkt j,\Vkt k$ are the quaternion units. The scalar part is $q_0$ and the pure part equals $q_1\Vkt i+q_2\Vkt j+q_3\Vkt k$, 
which is also denoted by $\frak{q}$. 
The conjugated quaternion to  $\frak{Q}=q_0+\frak{q}$ is given by 
$\widetilde{\frak{Q}}:=q_0-\frak{q}$ and $\frak{Q}$ is called a unit-quaternion for
$\frak{Q}\circ \widetilde{\frak{Q}}=1$,
where $\circ$ denotes the 
quaternion multiplication.  
Finally we can embed points $\go P$ of the Euclidean 3-space $E^3$ with coordinates $(p_1,p_2,p_3)$ 
with respect to the Cartesian frame $({\sf{O}} ; x_1, x_2, x_3)$ into the set of pure quaternions by: 
\begin{equation}\label{einbettung}
\iota_3:\RR^3\rightarrow \HH\quad \text{with}\quad (p_1,p_2,p_3)\,\,\mapsto\,\, \frak{p}:=p_1\Vkt i+p_2\Vkt j+p_3\Vkt k.
\end{equation}
Combining two quaternions by the dual unit $\varepsilon$ with the property $\varepsilon^2=0$ yields the set of  
dual quaternions $\HH+\varepsilon \HH$. 
An element $\frak{E}+\varepsilon\frak{T}$ of $\HH+\varepsilon \HH$
is called dual unit-quaternion if $\frak{E}$ is a unit-quaternion and the following condition holds:
\begin{equation}\label{study:con}
\frak{T}\circ \widetilde{\frak{Q}} + \frak{Q}\circ \widetilde{\frak{T}}=0 \quad \Longleftrightarrow\quad  e_0t_0+e_1t_1+e_2t_2+e_3t_3=0.
\end{equation}
We denote the set of dual unit-quaternions by $\EE$. By skipping the so-called Study condition of Eq.\ (\ref{study:con}) we get 
a superset of  $\EE$, which we call $\FF$. Note that both sets $\EE$ and $\FF$ 
build a group with respect to the quaternion multiplication.

Based on the usage of $\frak{E}+\varepsilon\frak{T}\in\EE$ 
it can be shown (e.g.\ \cite[ Section 3.3.2.2]{husty_sachs}) 
that the mapping of points $\go P\in E^3$ induced by any element of SE(3), can be written as:  
\begin{equation}\label{xstrich}
\delta_3:\HH\rightarrow \HH\quad \text{with}\quad
\frak p\,\,\mapsto\,\, \frak E\circ\frak p \circ \widetilde{\frak{E}} + 
(\frak T\circ \widetilde{\frak{E}}-\frak E\circ \widetilde{\frak{T}}).
\end{equation}
Moreover the mapping of Eq.\ (\ref{xstrich}) is an element of SE(3) for any 
$\frak{E},\frak{T}$ implying a dual quaternion $\frak{E}+\varepsilon\frak{T}\in\EE$. 
Note that $\delta_3(\frak p)$ is again a pure quaternion, 
where the first summand $\frak E\circ\frak p \circ \widetilde{\frak{E}}$ 
is the rotational component
and the term in the parentheses corresponds to the translational part.

As both dual quaternions $\pm(\frak{E}+\varepsilon\frak{T})\in \EE$ imply the same Euclidean motion of $E^3$, 
we consider the homogeneous 8-tuple $(e_0:\ldots:e_3:t_0:\ldots :t_3)$. 
These so-called Study parameters can be interpreted as points of a projective 7-dimensional space $P^7$. 
Therefore there is a bijection between SE(3) and all real points of $P^7$ located on the Study quadric 
$\Phi\subset P^7$ given by Eq.\ (\ref{study:con}), which  
is sliced along the 3-dimensional generator-space $G: \,\,e_0=e_1=e_2=e_3=0$.
Points of this generator-space are called "Pseudosomen" by Study \cite{study}.

\subsection{Motivation}

It is well known (e.g. \cite{pottmann_wallner}) that there exists a bijection between the set $\mathcal{L}$ of lines of the projective 3-space 
and all real points of the so-called Pl\"ucker  quadric 
\begin{equation}\label{pk_quadric}
\Psi: \quad l_{01}l_{23}+l_{02}l_{31}+l_{03}l_{12}=0
\end{equation}
of $P^5$, where the homogeneous 6-tuple $(l_{01}:l_{02}:l_{03}:l_{23}:l_{31}:l_{12})$ are the 
Pl\"ucker 
coordinates of the lines. This bijection $\mathcal{L}\rightarrow \Psi$ is also 
known as Klein mapping. 

\begin{remark}
Clearly there is an analogy between the Study quadric $\Phi$ and the  Pl\"ucker quadric $\Psi$. 
In contrast to SE(3) the set $\mathcal{L}$ is compact, i.e. the quadric $\Psi$ has not to be sliced.  
But the 2-dimensional generator space $l_{01}=l_{02}=l_{03}=0$ corresponds to the set 
of ideal lines, which therefore are the analog of the "Pseudosomen".  \hfill $\diamond$
\end{remark}

Moreover one can define an {\it extended Klein mapping} which identifies each point of $P^5$ with a  
linear complex of lines in $P^3$. As the latter can always be seen as a path-normal complex of   
instantaneous screws (different from the zero screw), which only vary in speed (i.e. they are real proportional), 
one ends up with a bijection between points of $P^5$ and {\it homogeneous screw coordinates} 
$\$:=(s_{01}:s_{02}:s_{03}:s_{23}:s_{31}:s_{12})$.  The corresponding path-normals fulfill the condition
\begin{equation}
s_{01}l_{23}+ s_{02}l_{31}+ s_{03}l_{12}+
s_{23}l_{01}+ s_{31}l_{02}+ s_{12}l_{03}= 0,
\end{equation}
which geometrically means that they are located in the intersection of $\Psi$ and the polar plane of $\$$ with respect to $\Psi$. 
This line/screw-geometric explanation of the complete $P^5$ raises the question for a kinematic interpretation of 
the points of the Study quadric's ambient space $P^7$, which is still missing. 

Moreover this study is motivated by the {\it extended inverse kinematic map} 
$\kappa^{-1}:\,\, P^7\setminus G \rightarrow SE(3)$ discussed by 
 Pfurner, Schr\"ocker and Husty (PSH) in \cite{psh}.
Under this map any point of $P^7\setminus G$ is identified with a displacement of SE(3), which corresponds to a point on $\Phi\setminus G$. 
These two points are linked by the so-called (cf.\ \cite{swc}) PSH-map: $P^7\setminus G \rightarrow  \Phi\setminus G$, 
which can be written in terms of dual quaternions as:
\begin{equation}\label{psh}
\phi:\FF\rightarrow \EE \quad\text{with}\quad \frak{E}+\varepsilon\frak{T} \,\,\mapsto\,\,  \frak{E}+\varepsilon\left[
\frak{T}-\tfrac{1}{2}\left(\frak{T}\circ \widetilde{\frak{E}} + \frak{E}\circ \widetilde{\frak{T}}
\right)\circ\frak{E}
\right].
\end{equation}
In this context it should also be mentioned that the PSH-map is the analogue of the map in $P^5$ sending 
a screw $\$$ to its axis (cf.\ \cite{swc} and \cite[Remark 10]{nawratil_ole}).
 
\begin{remark}
Selig et al.\ \cite[Theorem 1]{swc} showed that the PSH-map is equivalent to the composition of an extended 
inverse Cayley map with the direct Cayley map, where the Cayley map in question is associated to the adjoint 
representation of SE(3). Moreover it should be noted
that the basic principle of $\kappa^{-1}$ was already mentioned by Ge and Purwar \cite[Section 4]{ge_purwar}, who 
used it inter alia in \cite{dual_weights} and  \cite{purwar_ge}. 
\hfill $\diamond$
\end{remark}

We do not want to interpret points of the ambient space $P^7$ of $\Phi$ by means of $\kappa^{-1}$ as this map is not 
bijective (in contrast to the {\it extended Klein mapping}). We overcome this problem by interpreting the points of  $P^7$ 
as displacements of a motion group of the Euclidean 4-space $E^4$ (cf.\ Section \ref{sec:4space}). From this point of view 
we discuss the PSH-map (cf.\ Section \ref{sec:4space}), straight lines in $P^7$ (cf.\ Section \ref{sec:darboux}) and 
linear complexes of SE(3)-displacements  (cf.\ Section \ref{sec:lincomp}). 
We conclude the paper (cf.\ Section \ref{sec:con}) by noting an application for the interactive design of rational motions.

\section{Kinematic interpretation as a subgroup of SE(4)}\label{sec:4space}

We embed points $\go P$ of $E^4$ with coordinates $(p_0,p_1,p_2,p_3)$ 
with respect to the Cartesian frame $({\sf{O}} ; x_0,x_1,x_2,x_3)$ into the set of quaternions by 
the mapping:
\begin{equation}\label{einbettung4}
\iota_4:\RR^4\rightarrow \HH\quad \text{with}\quad (p_0,p_1,p_2,p_3)\,\,\mapsto\,\, \frak{P}:=p_0+p_1\Vkt i+p_2\Vkt j+p_3\Vkt k=
p_0+\frak{p}.
\end{equation}
Let us identify $E^3$ with any hyperplane $x_0=k$ with $k\in\RR$. Moreover we consider the set X$_4$ of displacements of SE(4), which 
fixes the ideal point in direction of the $x_0$-axis. 
It is an easy task to show that  X$_4$ is a subgroup\footnote{The 
nomenclature is due to the fact that  X$_4$ can be seen as the 4-dimensional analogue of a Schoenflies group of SE(3), 
which is denoted by X.} of SE(4). 

\begin{theorem}\label{proofX}
The mapping of points $\go P\in E^4$ induced by any element of  X$_4$, can be written as: 
\begin{equation}\label{schoenflies4}
\delta_4:\HH\rightarrow \HH\quad \text{with}\quad
\frak P\,\,\mapsto\,\, \frak E\circ\frak P \circ \widetilde{\frak{E}} - 2\frak E\circ \widetilde{\frak{T}}.
\end{equation}
Moreover the mapping of Eq.\ (\ref{schoenflies4}) is an element of  X$_4$ for any
$\frak{E},\frak{T}$ implying a dual quaternion $\frak{E}+\varepsilon\frak{T}\in \FF$. 
\end{theorem}

\begin{proof}
Direct computation shows that the direction of the $x_0$-axis remains fixed under the mapping of Eq.\ (\ref{schoenflies4}). 
Therefore any displacement of this form is an element of  X$_4$. It remains to show that any element of  X$_4$ can be written in this way. 
Therefore we use the following corollary which follows immediately from \cite[Theorem 3.1]{nawratil_fundamentals}: 
\begin{corollary}
The mapping $\frak P\,\,\mapsto\,\, \frak E\circ\frak P \circ \widetilde{\frak{E}}$
is a rotation about the plane $\Lambda$ through the origin $\sf O$ spanned by $\sf I:=(1,0,0,0)$ and ${\sf E}:=(e_0,e_1,e_2,e_3)$, 
where the rotation angle $\lambda$ is two times the enclosed angle of these two vectors; i.e. 
$0^{\circ} \leq \frac{\lambda}{2}\leq 180^{\circ}$. 
The angle of rotation $\lambda\geq 0$ is given with respect to the orientation of $\Lambda$, which is implied 
by the oriented triangle $\sf O, \sf I, \sf E$. 
\end{corollary}
The last sentence also implies that  $\frak E$ and $-\frak E$ yield the same rotation. Moreover from
this geometric interpretation it is clear that any rotation about a plane through $\sf O$ containing $\sf I$ 
can be written in the form $\frak P\,\,\mapsto\,\, \frak E\circ\frak P \circ \widetilde{\frak{E}}$. 

We proceed with the translation part $- 2\frak E\circ \widetilde{\frak{T}}$ of Eq.\ (\ref{schoenflies4}):  
As shown in \cite[Proof of Theorem 2.6]{nawratil_fundamentals} any translation vector of $\RR^4$ can be generated in this way. Moreover 
there is a bijection between $- 2\frak E\circ \widetilde{\frak{T}}$  and  
$\frak{T}$ for a given unit-quaternion $\frak{E}$. 

The final ingredient for our proof is the knowledge (e.g.\  \cite[Chapter 9]{berger} and \cite[Chapter I]{bottema}) 
that any displacement in $E^4$ (beside the identity transform) fixing the $x_0$-direction is either 
\begin{enumerate}[(a)]
\item
a translation or   
\item
a rotation about a plane $\Gamma$ parallel to the $x_0$-direction or
\item
a composition of a rotation about a plane $\Gamma$ parallel to the $x_0$-direction and a translation parallel to $\Gamma$.
\end{enumerate} 
Clearly, all translations (including the identity) can be written as in Eq.\ (\ref{schoenflies4}) by  setting $\frak{E}=1$. 
Moreover any displacements of item (b) and (c) can be archived by the composition of a rotation about a plane $\Lambda$ parallel to $\Gamma$, 
which contains the origin $\sf O$, and a translation. If the translation vector is orthogonal to $\Lambda$ then we obtain 
case (b) otherwise we end up with case (c). This closes the proof of Theorem \ref{proofX}. \hfill $\BewEnde$ 
\end{proof}

As both dual quaternions $\pm(\frak{E}+\varepsilon\frak{T})\in\FF$  
correspond to the same  X$_4$-motion of $E^4$, we consider again the homogeneous 8-tuple $(e_0:\ldots:e_3:t_0:\ldots :t_3)$. 
Therefore there is a bijection between  X$_4$ and all real points of $P^7\setminus G$. 
Due to 
\begin{equation}
\delta_4(\frak P) =
p_0-2(e_0t_0+e_1t_1+e_2t_2+e_3t_3) + \delta_3(\frak p)
\end{equation}
the mapping of Eq.\ (\ref{schoenflies4}) restricted to the {\it pure part} equals the mapping 
of Eq.\ (\ref{xstrich}). Moreover displacements of  X$_4$ fixing the hyperplanes $x_0=k$ correspond to points on 
$\Phi\setminus G$. These points imply SE(3)-displacements in the hyperplanes 
$x_0=k$, which we identify with $E^3$. This completes the kinematic interpretation of all points of $P^7\setminus G$; i.e.
the Study parameter space with exception of the "Pseudosomen".

\begin{remark}
For reasons of completeness we also want to give the normalized Grassmann plane coordinates of the plane $\Gamma$, 
which has been mentioned in item (b) and (c) of the proof above. 
According to \cite[Section 4.3]{nawratil_fundamentals} they can be written as a quaternionic triple 
$(\overline{\frak l},\dach{\frak l}, \frak L)$ where $(\overline{\frak l},\dach{\frak l})$ correspond to the normalized
Pl\"ucker coordinates of the ideal line of $\Gamma$, oriented from the ideal point in direction of $\sf I$ to the ideal 
point in direction of $\sf E$. A short computation shows that $\Gamma$ is determined by:
\begin{equation}
\overline{\frak l}:=\frac{\frak e}{\frak e\circ\widetilde{\frak{e}}} \qquad
\dach{\frak l}:=0 \quad\text{and}\quad
\frak L:=\frac{(\frak E-  \widetilde{\frak{E}})\circ (\overline{\frak l}\circ\frak T-\frak T\circ\overline{\frak l})}
{(\frak E-  \widetilde{\frak{E}})\circ(\widetilde{\frak{E}}-\frak E)}.
\end{equation}
Note that the intersection of the hyperplane $x_0=k$ and the plane $\Gamma$ 
yields the axis of the corresponding displacement (pure rotation or screw displacement) in $E^3$. 
\hfill $\diamond$
\end{remark}

Now we want to clarify the meaning of the extended inverse kinematic mapping $\kappa^{-1}$ in terms of our given 
interpretation. Therefore we compute the action of $\phi(\frak{E}+\varepsilon\frak{T})\in\EE$ on a 
point $\go P\in E^4$ according to Eq.\ (\ref{schoenflies4}), which yields: 
\begin{equation}
\frak P\,\,\mapsto\,\, \frak E\circ\frak P \circ \widetilde{\frak{E}} - 2\frak E\circ \widetilde{\frak{T}} + 
(\frak T\circ \widetilde{\frak{E}}+\frak E\circ \widetilde{\frak{T}}).
\end{equation}
The term in the parentheses is a pure scalar which equals the 
scalar part of $2\frak E\circ \widetilde{\frak{T}}$. Therefore the inverse kinematic mapping 
is nothing else than an orthogonal projection onto the hyperplane $x_0=k$, which we have identified with $E^3$. 

In order to make the things more descriptive we stress the following lower-dimensional analogue:  Consider the 
Schoenflies motion group X with respect to the direction $x_3$ and its subgroup of planar motions SE(2) 
parallel to the $x_1x_2$-plane. If we apply an  orthogonal projection along the $x_3$-direction (analogue of $\kappa^{-1}$) 
to a X-motion we obtain a planar one.  This is illustrated in Fig.\ \ref{fig1}.

\begin{figure}[b]
\begin{center} 
\begin{overpic}
    [height=30mm]{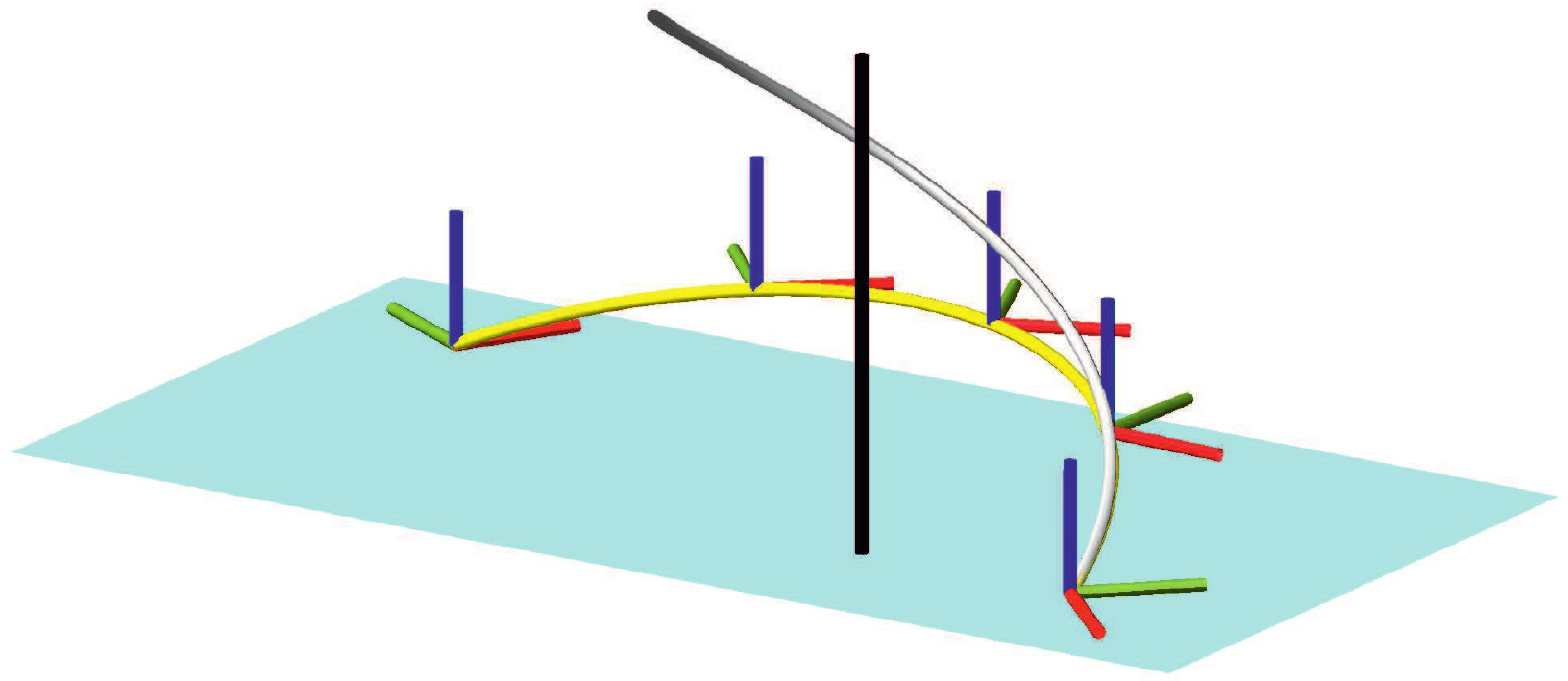}
		\begin{scriptsize}
		\put(6.5,14.5){$x_1x_2$-plane}
\end{scriptsize}         
  \end{overpic} 
	\hfill
 \begin{overpic}
    [height=25.8mm]{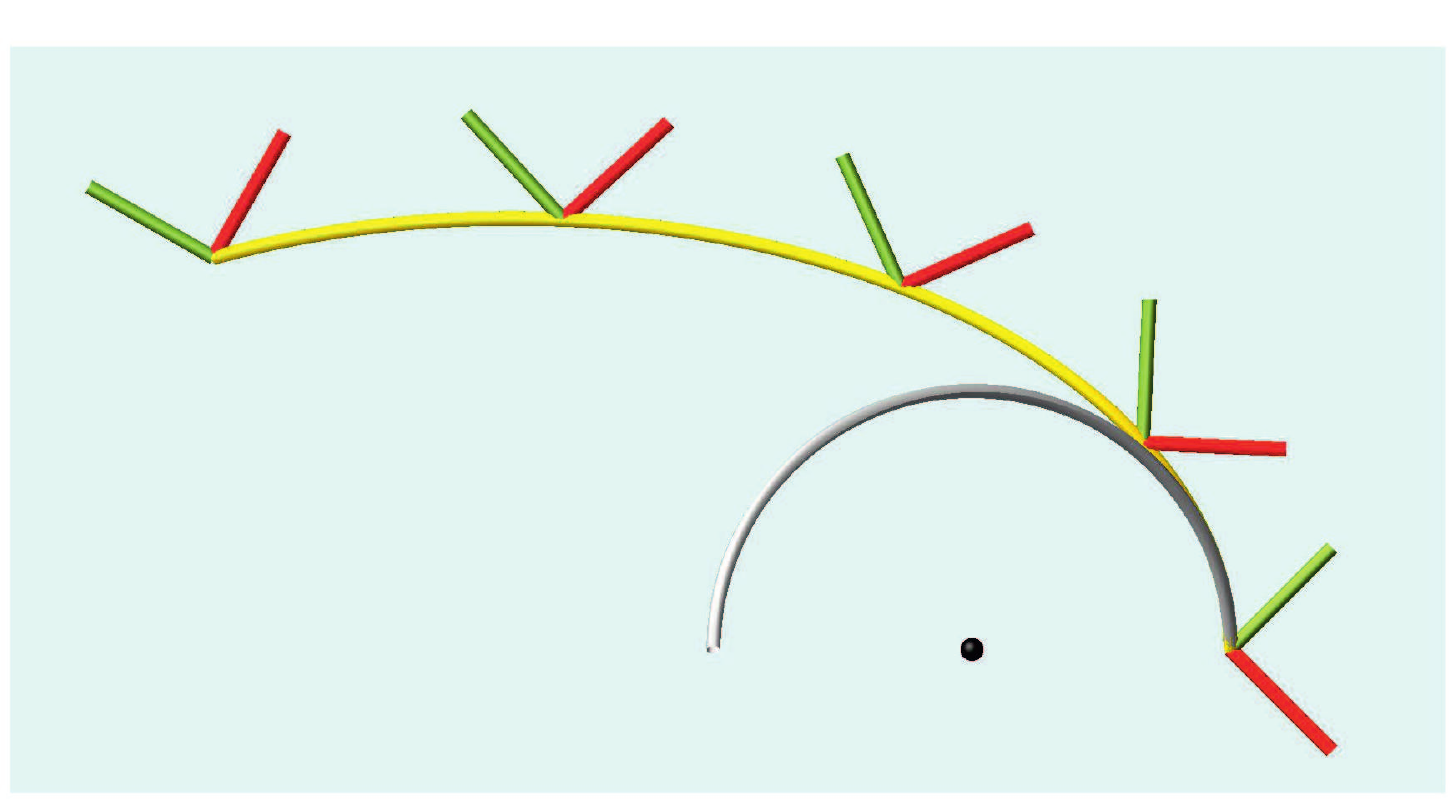}
\begin{scriptsize}
\put(88,8.5){$(0)$}
\put(80,28.5){$(\tfrac{3}{4})$}
\put(60.9,40){$(1)$}
\put(34.9,45.5){$(\tfrac{3}{4})$}
\put(9.5,42.5){$(0)$}
\end{scriptsize}       
  \end{overpic} 
\end{center} 
\caption{The planar motion (right) is obtained as top view of the Schoenflies motion (left). 
In addition we label the poses of the top view by the $x_3$-coordinate. In German such a map is  
known as "{\it kotierte Projektion}". \newline
The instantaneous screw in the starting pose (left), which is visualized by a gray helix and the black axis in  
$x_3$-direction, appears in the top view (right) as the instantaneous rotation of the planar motion. 
 } 
\label{fig1}
\end{figure}

\begin{remark}
This lower-dimensional analogue is just for illustration but would not make sense in practice. 
The reason is that SE(2) is isomorphic to $P^3$ with coordinates $(e_0:e_3:t_1:t_2)$ 
sliced along a line $e_0=e_3=0$, but X is isomorphic to a hyperquadric $e_0t_0+e_3t_3=0$ of $P^5$ with coordinates 
$(e_0:e_3:t_0:t_1:t_2:t_3)$, which is sliced along the 3-plane $e_0=e_3=0$. Therefore the superset X of SE(2) has the "nasty"
quadratic side-condition, which is exactly opposite for X$_4$ and SE(3). \hfill $\diamond$
\end{remark}

\newpage
 
\section{Straight lines in the Study parameter space}\label{sec:darboux}

It is known (cf.\ \cite[Section 3.2]{purwar_ge}, \cite[Section 4]{swc}, \cite[Section 3]{schrocker})  
that straight lines in $P^7$ are in general send by the PSH-map to vertical Darboux motions (see e.g.\ 
\cite[pp.\ 321--322]{bottema} and \cite{krames}). 
In special cases they correspond to rotations about fixed axes or translations along fixed directions.

As preparatory work for the study of straight lines of $P^7$ in terms of X$_4$-motions, we recall some 
results on so-called Darboux 2-motions, which were studied by Karger \cite{karger_darboux} in $E^n$. 
In  \cite[Definition 1]{karger_darboux} he defined a Darboux 2-motion as a motion, where all trajectories are 
(1) planar and (2) affine equivalent in the kinematical sense. The second condition is redundant 
for $E^4$ due to a result given by Karger in \cite{karger_4space}. 
Therefore all motions of SE(4) with planar trajectories, which are known \cite[Theorem 4]{karger_darboux} to be 
either ellipses or line segments, are given in \cite{karger_darboux}. 
We focus on Darboux 2-motions in $E^4$, where all points have circular trajectories.

\begin{theorem}\label{lem1}
A circular Darboux 2-motion, which is neither spherical nor a pure translation, has to be composed of a rotation about a fixed plane and 
a circular translations parallel to this plane. 
\end{theorem}

\begin{proof}
In the last remark of \cite{karger_darboux}, Karger mentioned that there are two types of Darboux 2-motions in $E^4$. 
In the following we study them separately:
\begin{enumerate}[1.]
\item
{\bf Type 1} {\it (characterized by Karger's parameters $k_1=1$, $\delta_1=\delta_2=1$; cf.\ \cite{karger_darboux}):} 
According to \cite[Theorem 2]{karger_darboux} this motion can be written in matrix form
(adapted in notation) for a proper choice of coordinate frames as: 
\begin{equation}\label{karger}
\begin{pmatrix}
p_0' \\ p_1' \\p_2' \\ p_3'
\end{pmatrix}
=
\begin{pmatrix}
1 & 0 & 0 & 0 \\
0 & 1 & 0 & 0 \\
0 & 0 & \cos(\tau) & -\sin(\tau) \\
0 & 0 & \sin(\tau) & \cos(\tau) \\
\end{pmatrix}
\begin{pmatrix}
p_0 \\ p_1 \\p_2 \\ p_3
\end{pmatrix}
+
\begin{pmatrix}
\tfrac{\beta}{2\sqrt{\gamma}}\left(\gamma\sin(\tau)+1-\cos(\tau)\right) \\
\tfrac{\beta}{2\sqrt{\gamma}}\sqrt{\gamma^2-1}\left(1-\cos(\tau)\right) \\
0 \\
\tfrac{\nu}{2}\left(\cos(\tau)-1\right)
\end{pmatrix},
\end{equation}
where $\beta,\gamma,\nu\in\RR$ with $\beta\geq 0$, $\gamma\geq 1$, $\nu\geq 0$ and $\tau$ denoting the motion parameter. 
We apply the half-angle substitution $\cos(\tau)=\frac{1-t^2}{1+t^2}$ and  $\sin(\tau)=\frac{2t}{1+t^2}$ 
and introduce homogenous coordinates $p_i':=\frac{h_0'}{h'}$ for $i=0, \ldots ,3$. Now we intersect the 
trajectories with the hyperplane at infinity $h'=0$, which implies $t=\pm I$ where $I$ denotes the complex unit. 
Then it can easily be seen that $\nu=0$  
is a necessary condition for these two intersection points to be located 
on the absolute sphere $h_0'^2+h_1'^2+h_2'^2+h_3'^2=0$ independent of the point $(p_0,p_1,p_2,p_3)$. 
As a consequence the rotation plane ($x_0x_1$-plane) is fixed.  Then points of this plane can only have circular trajectories 
if  the translation is circular.
\item
{\bf Type 2} {\it (characterized by Karger's parameters $k_1=2$, $\delta_1=\delta_2=0$; cf.\ \cite{karger_darboux}):}
According to \cite[Theorem 2]{karger_darboux} this motion can be composed of a spherical motion about the origin $\sf O$, 
where the orthogonal matrix is given in \cite[Eq.\ (26)]{karger_darboux}, and a translational vector of the form $(0,0,\star(\tau),0)^T$ 
(cf.\  \cite[Eq.\ (27)]{karger_darboux}). 
Therefore the origin can only have a linear trajectory or it is fixed ($\Rightarrow$ spherical motion) 
if the expression $\star(\tau)$ vanishes independently of $\tau$. This finishes the proof. 
\hfill $\BewEnde$
\end{enumerate}
\end{proof}

\begin{remark}
On the first sight it is not clear how the remaining parameters $\beta\geq 0$ and $\gamma\geq 1$ of Type 1 have to be chosen 
such that we get a circular translation. 
A short computation shows that if we set $\beta:=c/\sqrt{\gamma}$ with $c\geq 0$ and then applying the 
limit $\gamma\rightarrow\infty$ to the translation vector of Eq.\ (\ref{karger}) yields the following expression:  
\begin{equation}\label{trans_part}
c\left(\sin(\tau),1-\cos(\tau),0,0\right)^T.
\end{equation}
Now it can easily be verified that the trajectory of each point $\go P\in E^4$ is circular for $c>0$. 
Up to the best knowledge of the author, the existence of these circular Darboux 2-motions 
has not been mentioned in the literature until now. 
Moreover they belong to the class of Borel-Bricard motions of $E^4$ (i.e. all points 
have hyperspherical trajectories). Further examples of this class are given in \cite{vogler,wunderlich}. \hfill $\diamond$
\end{remark}

\begin{theorem}
A straight line $\in P^7\setminus G$ corresponds to one of the following X$_4$-motions and vice versa:
\begin{enumerate}[(i)]
\item
a translation along a fixed direction,
\item
a rotation about a fixed plane, 
\item
a circular Darboux 2-motion, which is neither spherical nor a pure translation.  
\end{enumerate} 
\end{theorem} 

\begin{proof}
Given are two different displacements of X$_4$ by $\frak{E}_1+\varepsilon\frak{T}_1\in \FF$ and  $\frak{E}_2+\varepsilon\frak{T}_2\in \FF$ 
which span a straight line in $P^7\setminus G$. The corresponding quaternion $\frak E(s)+\varepsilon \frak T(s)\in \FF$ of the motion can be written 
in dependence of the parameter $s$ as:
\begin{equation}
\frak E(s)+\varepsilon \frak T(s):=
\frac{\frak E_s}{\sqrt{\frak E_s\circ\widetilde{\frak{E}_s}}} + \varepsilon \frac{\frak T_s}{\sqrt{\frak E_s\circ\widetilde{\frak{E}_s}}} 
\quad\text{with}\quad
\begin{cases}
\frak E_s:=s\frak E_1+(1-s)\frak E_2 \\
\frak T_s:=s\frak T_1+(1-s)\frak T_2
\end{cases}
\end{equation}
If the given two displacements have the same orientation ($\frak{E}_1=\pm \frak{E}_2$) then the motion is trivially a 
pure translation along a fixed direction; i.e.\ case (i). Therefore we can assume $\frak{E}_1\neq\pm \frak{E}_2$ for the remainder of the proof. 

Moreover if the expression $\frak T(s)\circ \widetilde{\frak{E}}(s)+\frak E(s)\circ \widetilde{\frak{T}}(s)$ is constant $-k$ for all $s\in\RR$, 
then the relative motion between each two poses of this motion is located on the Study quadric. Therefore 
the SE(3)-motion in the hyperplane $x_0=k$ has to be a pure rotation about a fixed line. As a consequence we get case (ii). 

If we also exclude this scenario we are left with the general case. We can compute the trajectories by means of 
Eq.\ (\ref{schoenflies4}). As the obtained coordinates are rational quadratic polynomials in $s$, the trajectories 
are conic sections. It remains to show that they are circles, which can be done analogously to Theorem \ref{lem1}; i.e. 
both intersection points with the hyperplane at infinity have to be located on the absolute sphere. This 
can be verified by e.g.\ direct computations with {\sc Maple}. 

Now we tackle the reverse direction, that every  X$_4$-motion of type (i), (ii) or (iii) corresponds to a straight line in  $P^7\setminus G$.  
For (i) and (ii) this is trivial and left to the reader. For (iii) we proceed as follows: Without loss of generality we can choose a reference frame 
in a way that the rotation plane is the $x_0x_1$-plane. This corresponds to the rotation matrix of Eq.\ (\ref{karger}). But as we 
cannot chose the $x_0$-direction freely (it is predetermined) 
the translation vector of Eq.\ (\ref{trans_part}) has to be rotated about the 
$x_2x_3$-plane by the angle $\varrho$, which yields:
\begin{equation*}
c\left(\cos(\varrho)\sin(\tau)-\sin(\varrho)(1-\cos(\tau)),
\sin(\varrho)\sin(\tau)+\cos(\varrho)(1-\cos(\tau)),0,0\right)^T.
\end{equation*}
It can be checked that the corresponding straight line in $P^7\setminus G$ is obtained by:
\begin{equation*}
\frak E_1=1, \quad \frak T_1=0, \quad
\frak E_2=\Vkt i, \quad \frak T_2=-c\cos(\varrho)+c\sin(\varrho)\Vkt i 
\quad \text{and} \quad s=\frac{1}{1+t}.
\end{equation*}
As a change of the reference frame within the group X$_4$ implies a linear transformation of the 
corresponding motion parameters in $P^7$ ($\Rightarrow$ a straight line is mapped to a straight line), the proof is completed. 
\hfill $\BewEnde$
\end{proof}

\section{Linear complex of SE(3)-displacements}\label{sec:lincomp}

As already discussed in the proof of Theorem \ref{proofX} 
every X$_4$-displacement determined by $\frak{E}+\varepsilon\frak{T}\in \FF$ is a composition of 
a rotation about a plane $\Gamma$ parallel to the plane spanned by $\sf I$ and $\sf E$ and 
a translation vector $\sf C$ parallel to $\Gamma$. 
Therefore we can consider a further geometric quantity; namely the angle enclosed by $\sf E$ and $\sf C$. 
For orthogonal vectors  $\sf E$ and $\sf C$ we call the X$_4$-displacement {\it orthogonal}. 
This is the case if and only if the translation part $- 2\frak E\circ \widetilde{\frak{T}}$ of Eq.\ (\ref{schoenflies4})
is orthogonal to $\sf E$ which is equivalent with the fact that $\frak{T}$ is a pure quaternion. This yields
the following theorem:

\begin{theorem}\label{ortho}
A X$_4$-displacement defined by $\frak{E}+\varepsilon\frak{T}\in \FF$ is orthogonal if and only if $\frak{T}$ is a pure quaternion.
\end{theorem}

Given is a displacements of X$_4$ by $\frak{E}+\varepsilon\frak{T}\in \FF$. Now we want to compute all  
X$_4$-displacements fixing the hyperplanes $x_0=k$ (SE(3)-displacements given by $\frak{F}_i+\varepsilon\frak{U}_i\in \EE$), 
which are located in the polar plane of $\frak{E}+\varepsilon\frak{T}\in \FF$ with respect to the Study quadric $\Phi$. 
The corresponding condition for this so-called {\it linear complex of SE(3)-displacements} (with respect to $\frak{E}+\varepsilon\frak{T}\in \FF$)
can be written as:
\begin{equation}\label{lincon}
(\frak{F}_i\circ\widetilde{\frak{T}}+\frak{T}\circ\widetilde{\frak{F}}_i) + 
(\frak{E}\circ\widetilde{\frak{U}}_i+\frak{U}_i\circ\widetilde{\frak{E}}) = 0. 
\end{equation}
Then the relative motion $\frak{G}_i+\varepsilon\frak{V}_i\in \FF$ transforming the pose defined by $\frak{F}_i+\varepsilon\frak{U}_i\in \EE$ into the 
pose, which corresponds to the pole $\frak{E}+\varepsilon\frak{T}\in \FF$, is given by:
\begin{equation}\label{relativ}
\frak G_i:=\frak E\circ\widetilde{\frak{F}}_i \quad\text{and}\quad
\frak{V}_i:=\frak E\circ \widetilde{\frak{U}}_i + \frak T\circ \widetilde{\frak{F}}_i.
\end{equation}
Now Eq.\ (\ref{lincon}) implies that $\frak{V}_i$ is a pure quaternion. Under consideration of Theorem \ref{ortho}
we get the following result.

\begin{theorem}
The linear complex of SE(3)-displacements with respect to a given X$_4$-displacement $\frak{E}+\varepsilon\frak{T}\in \FF$ 
consists of all points of the Study quadric $\Phi\setminus G$, which correspond to $\frak{F}_i+\varepsilon\frak{U}_i\in \EE$, such that  
the relative motion  $\frak{G}_i+\varepsilon\frak{V}_i\in \FF$ of Eq.\ (\ref{relativ}) is an orthogonal X$_4$-displacement. 
\end{theorem}
According to the sentence below Eq.\ (\ref{psh}) we call 
the displacement given by $\phi(\frak{E}+\varepsilon\frak{T})$ 
the axis of the linear complex of SE(3)-displacements. 

\begin{remark} 
If $\frak{E}+\varepsilon\frak{T}\in \EE$ holds; i.e.\ the corresponding point is located on the Study quadric $\Phi\setminus G$, 
then the related linear complex of SE(3)-displacements consists of all poses, which can be generated from the  
pole displacement ($=$ axis displacement) either by pure rotations about lines of $E^3$ or by pure translations. 
\hfill $\diamond$
\end{remark}

\section{Conclusion}\label{sec:con}

The given kinematic interpretation of all points of the Study quadric's ambient space  $P^7$ 
in terms of X$_4$-displacements can e.g.\ be used for the interactive design of rational motions in the following way (see Fig.\ \ref{fig2}): 

We start with a {\it projective de Casteljau construction} in $P^7$. 
Due to the given interpretation the control points and Farin points of this construction in $P^7$  
correspond to poses in $E^4$, which can be projected orthogonally along the $x_0$-direction 
onto $E^3$. In addition we label the obtained poses of $E^3$ by the $x_0$-coordinate ("{\it kotierte Projektion}"; cf.\ Fig.\ \ref{fig1}). 
As the resulting Farin poses and control poses in $E^3$ 
are not affinely distorted (as done by some other known methods \cite{wagner,roeschel}), 
the user can modify very intuitively the control structure (Farin and control poses as 
well as their $x_0$-heights).  

Moreover, as a meaningful optimization criterion of the resulting rational motion one can use the minimization of the 
maximal translation along the $x_0$-direction during the motion.

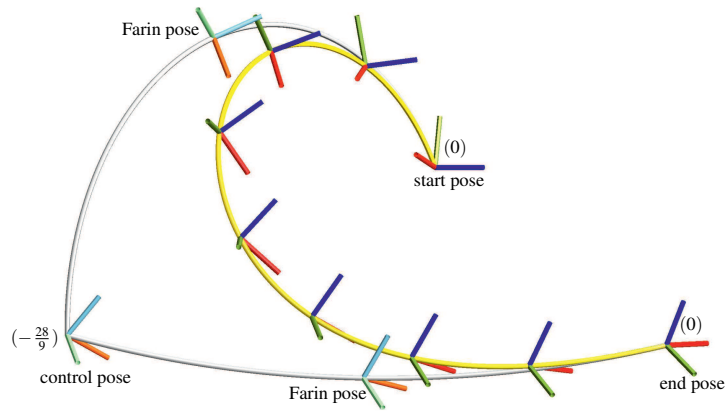
\begin{figure}[b]
\begin{center} 
 \begin{overpic}
    [height=55mm]{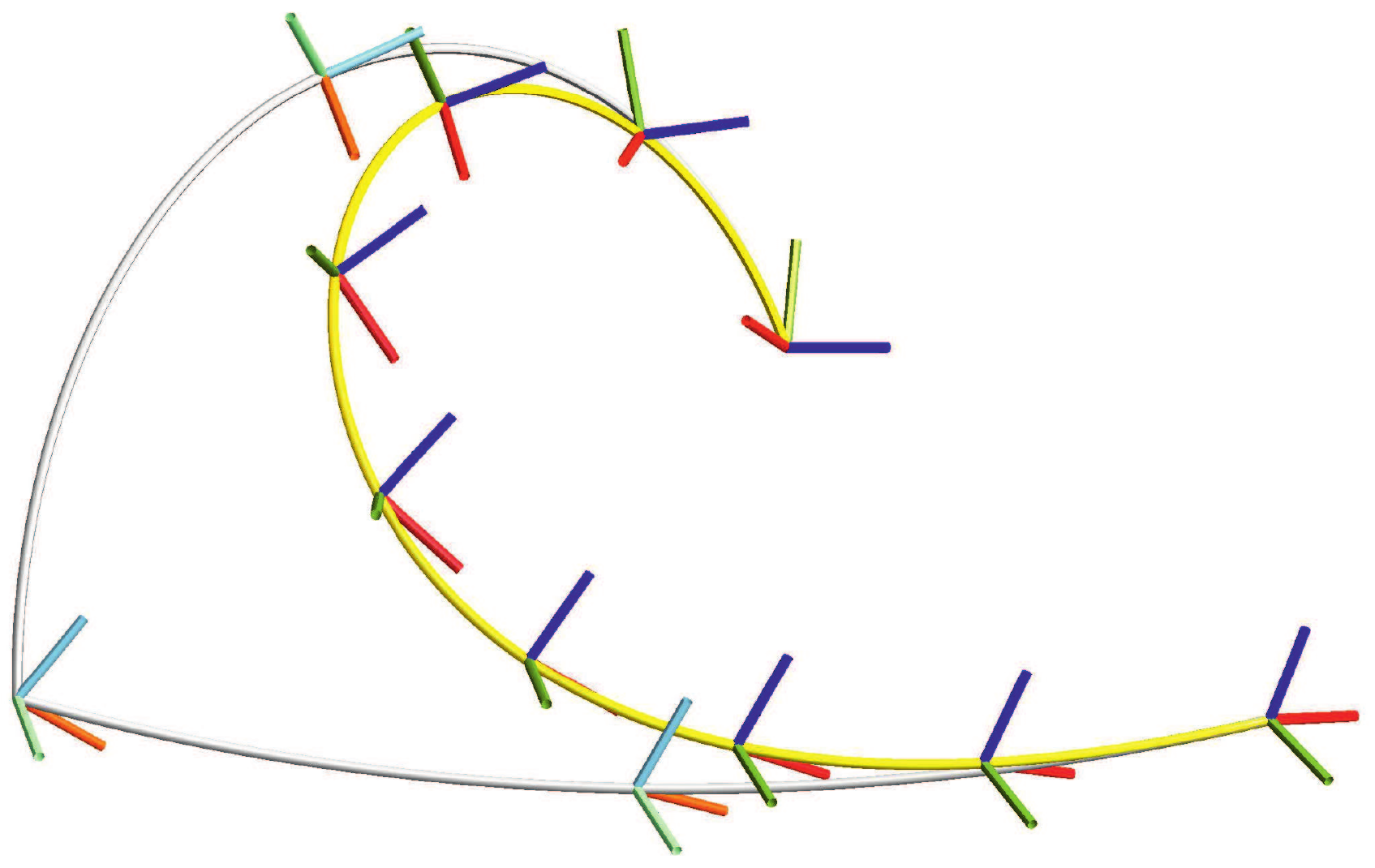}
		\begin{scriptsize}
\put(94.5,13){$(0)$}
\put(91.5,4){end pose}
\put(58.5,40){$(0)$}
\put(54,35){start pose}
\put(-7.5,11){$(-\tfrac{28}{9})$}
\put(-3,5){control pose}
\put(9.5,58){Farin pose}
\put(35,2.5){Farin pose}
\end{scriptsize} 
  \end{overpic} 
\end{center} 
\caption{The illustrated quartic rational motion corresponds to a quadratic Bezier curve in $P^7$. Each {\it Farin pose} 
can only be modified within the vertical Darboux motion determined by the {\it control pose} and {\it start/end pose}, respectively. The related two 
trajectories of the origin are illustrated as gray elliptic arcs. In contrast the {\it control pose} has $7$ degrees of freedom, 
which can be used for motion design. The $x_0$-coordinates of the control, start and end pose are given in parentheses. 
 } 
\label{fig2}
\end{figure}

\begin{acknowledgement}
The author is supported by Grant No.~P~24927-N25 of the Austrian Science Fund FWF within the
project "Stewart Gough platforms with self-motions". 
\end{acknowledgement}

\end{document}